\documentclass[11pt]{article}

\usepackage[letterpaper,margin=0.95in]{geometry}
\usepackage{amsmath, amssymb, amsthm, amsfonts}
\usepackage{bbm}

\usepackage{ifthen}
\usepackage{tikz}
\usetikzlibrary{positioning,decorations.pathreplacing}

\usepackage{cite}
\usepackage{appendix}
\usepackage{graphicx}
\usepackage{color}
\usepackage{algorithm}
\usepackage{algorithm}
\usepackage[noend]{algpseudocode}
\usepackage{epstopdf}
\usepackage{wrapfig}
\usepackage{paralist}
\usepackage[textsize=tiny]{todonotes}

\usepackage{framed}
\usepackage[framemethod=tikz]{mdframed}
\usepackage[bottom]{footmisc}
\usepackage{enumitem}
\setitemize{noitemsep,topsep=3pt,parsep=3pt,partopsep=3pt}
\usepackage[font=small]{caption}
\usepackage{xspace}

\newtheorem{theorem}{Theorem}[section]
\newtheorem{lemma}[theorem]{Lemma}
\newtheorem{meta-theorem}[theorem]{Meta-Theorem}
\newtheorem{claim}[theorem]{Claim}

\definecolor{darkgreen}{rgb}{0,0.5,0}
\usepackage{hyperref}
\hypersetup{
    unicode=false,          
    colorlinks=true,        
    linkcolor=red,          
    citecolor=darkgreen,        
    filecolor=magenta,      
    urlcolor=cyan           
}
\usepackage[capitalize]{cleveref}

\algnewcommand\algorithmicswitch{\textbf{switch}}
\algnewcommand\algorithmiccase{\textbf{case}}

\algdef{SE}[SWITCH]{Switch}{EndSwitch}[1]{\algorithmicswitch\ #1\ \algorithmicdo}{\algorithmicend\ \algorithmicswitch}%
\algdef{SE}[CASE]{Case}{EndCase}[1]{\algorithmiccase\ #1}{\algorithmicend\ \algorithmiccase}%
\algtext*{EndSwitch}%
\algtext*{EndCase}%

\newcommand{\eps}{\varepsilon}

\newcommand{\poly}{\operatorname{\text{{\rm poly}}}}

\renewcommand{\paragraph}[1]{\vspace{0.15cm}\noindent {\bf #1}:}

\newcommand{\FullOrShort}{full}

\ifthenelse{\equal{\FullOrShort}{full}}{
	
  \newcommand{\fullOnly}[1]{#1}
  \newcommand{\shortOnly}[1]{}
  
  }{

    \newcommand{\fullOnly}[1]{}
    \newcommand{\IncludePictures}[1]{}
   
  }

\begin{document}

\date{}

\title{An Improved Distributed Algorithm for \\Maximal Independent Set}

\author{
 Mohsen Ghaffari\\
  \small MIT \\
  \small ghaffari@mit.edu
 }

\maketitle


\begin{abstract}
The Maximal Independent Set (MIS) problem is one of the basics in the study of \emph{locality} in distributed graph algorithms. This paper
presents an extremely simple randomized algorithm providing a near-optimal \emph{local complexity} for this problem, which incidentally, when combined with some known techniques, also leads to a near-optimal \emph{global complexity}.

\smallskip
Classical MIS algorithms of Luby [STOC'85] and Alon, Babai and Itai [JALG'86] provide the \emph{global complexity} guarantee that, with high probability\footnote{As standard, we use the phrase \emph{with high probability} to indicate that an event has probability at least $1-1/n$.}, \emph{all nodes} terminate after $O(\log n)$ rounds. In contrast, our initial focus is on the \emph{local complexity}, and our main contribution is to provide a very simple algorithm guaranteeing that \emph{each} particular node $v$ terminates after $O(\log \mathsf{deg}(v)+\log 1/\eps)$ rounds, with probability at least $1-\eps$. 
The guarantee holds even if the randomness outside $2$-hops neighborhood of $v$ is determined adversarially. 
This degree-dependency is optimal, due to a lower bound of Kuhn, Moscibroda, and Wattenhofer [PODC'04].  

\smallskip  
Interestingly, this local complexity smoothly transitions to a global complexity: by adding techniques of Barenboim, Elkin, Pettie, and Schneider [FOCS'12; arXiv: 1202.1983v3], we\footnote{\emph{quasi nanos, gigantium humeris insidentes}} get a randomized MIS algorithm with a high probability global complexity of $O(\log \Delta) + 2^{O(\sqrt{\log \log n})}$, where $\Delta$ denotes the maximum degree. This improves over the $O(\log^2 \Delta) + 2^{O(\sqrt{\log \log n})}$ result of Barenboim et al., and gets close to the $\Omega(\min\{\log \Delta, \sqrt{\log n}\})$ lower bound of Kuhn et al.

\smallskip
Corollaries include improved algorithms for MIS in graphs of upper-bounded arboricity, or lower-bounded girth, for Ruling Sets, for MIS in the Local Computation Algorithms (LCA) model, and a faster distributed algorithm for the Lov\'{a}sz Local Lemma. 

\end{abstract}

\vspace*{-7mm}
\setcounter{page}{0}
\thispagestyle{empty}

\newpage
\section{Introduction and Related Work}\label{sec:intro}
\vspace{-5pt}
Locality sits at the heart of distributed computing theory and is studied in the medium of problems such as Maximal Independent Set (MIS), Maximal Matching (MM), and Coloring. Over time, MIS has been of special interest as the others reduce to it. The story can be traced back to the surveys of Valiant\cite{valiant1983parallel} and Cook\cite{cook1983overview} in the early 80's which mentioned MIS as an interesting problem in non-centralized computation, shortly after followed by (poly-)logarithmic algorithms of Karp and Wigderson\cite{KarpWigderson}, Luby\cite{luby1985simple}, and Alon, Babai, and Itai\cite{alon1986fast}. Since then, this problem has been studied extensively. We refer the interested reader to \cite[Section 1.1]{barenboim2012locality}, which provides a thorough and up to date review of the state of the art.

In this article, we work with the standard distributed computation model called $\mathsf{LOCAL}$\cite{peleg:2000}: the network is abstracted as a graph $G=(V, E)$ where $|V|=n$; initially each node only knows its neighbors; communications occur in synchronous rounds, where in each round nodes can exchange information only with their graph neighbors. 

In the $\mathsf{LOCAL}$ model, besides it's practical application, the distributed computation time-bound has an intriguing purely graph-theoretic meaning: it identifies the radius up to which one needs to look to determine the output of each node, e.g., its color in a coloring. For instance, results of \cite{luby1985simple, alon1986fast} imply that looking only at the $O(\log n)$-hop neighborhood suffices, w.h.p.

\subsection{Local Complexity}
\label{subsec:localIntro}
\vspace{-5pt}
Despite the local nature of the problem, classically the main focus has been on the global complexity, i.e., the time till all nodes terminate. Moreover, somewhat strikingly, the majority of the standard analysis also take a non-local approach: often one considers the whole graph and shows guarantees on how the algorithm makes a \emph{global progress} towards it \emph{local objectives}. A prominent example is the results of \cite{luby1985simple, alon1986fast} where the analysis shows that per round, in expectation, half of the edges of the whole network get removed\footnote{These analysis do not provide any uniformity guarantee for the removed edges.}, hence leading to the \emph{global complexity} guarantee that after $O(\log n)$ rounds, with high probability, the algorithm terminates everywhere. 
See \Cref{warmup}.

This issue seemingly suggests a gap in our understanding of \emph{locality}. The starting point in this paper is to question whether this global mentality is necessary for obtaining the \emph{tight} bound\footnote{Without insisting on tightness, many straightforward (but weak) complexities can be given using local analysis.}. That is, can we instead provide a tight bound using \emph{local analysis}, i.e., an analysis that only looks at a node and some small neighborhood of it? To make the difference more sensible, let us imagine $n\rightarrow \infty$ and seek time-guarantees independent of $n$. 

Of course this brings to mind locality-based lower bounds which at first glance can seem to imply a negative answer: Linial\cite{linial1992locality} shows that even in a simple \emph{cycle} graph, MIS needs $\Omega(\log^* n)$ rounds, and Kuhn, Moscibroda and Wattenhofer\cite{kuhn2004localLB} prove that it requires $\Omega(\sqrt{\log n})$ rounds in some well-crafted graphs. But there is a catch: these lower bounds state that the time till all nodes terminate is at least so much. One can still ask, what if we want a time-guarantee for \emph{each single node} instead of \emph{all nodes}? While in the deterministic case these time-guarantees, called respectively \emph{local} and \emph{global} complexities, are equivalent, they can differ when the guarantee that is to be given is probabilistic, as is usual in randomized algorithms. Note that, the local complexity is quite a useful guarantee, even on its own. For instance, the fact that in a cycle, despite Linial's beautiful $\Omega(\log^* n)$ lower bound, the vast majority of nodes are done within $O(1)$ rounds is a meaningful property and should not be ignored. To be concrete, our starting question now is:
\vspace{3pt}
\begin{mdframed}[hidealllines=false,backgroundcolor=gray!10]
\vspace{-1pt}
\textbf{Local Complexity Question}: How long does it take till each particular node $v$ terminates, and knows whether it is in the (computed) MIS or not, with probability at least $1-\eps$?
\vspace{-2pt}
\end{mdframed}

Using $\Delta$ to denote the maximum degree, one can obtain answers such as $O(\log^{2} \Delta+ \log 1/\eps)$ rounds for Luby's algorithm, or $O(\log \Delta  \log \log \Delta$ $+ \log \Delta \log 1/\eps)$ rounds for the variant of Luby's used by Barenboim, Elkin, Pettie, and Schneider\cite{barenboim2012locality} and Chung, Pettie, and Su~\cite{chung2014LLL}. However, both of these bounds seem to be off from the right answer; e.g., one cannot recover from these the standard $O(\log n)$ high probability global complexity bound. In the first bound, the first term is troublesome and in the latter, the second term becomes the bottleneck. In both, the high probability bound becomes $O(\log^2 n)$ when one sets $\Delta=n^\delta$ for a constant $\delta>0$.

We present an extremely simple algorithm that overcomes this problem and provides a local complexity of $O(\log \Delta+ \log 1/\eps)$. More formally, we prove that:

\begin{theorem}
\label{thm:local} There is a randomized distributed MIS algorithm for which, for each node $v$, the probability that $v$ has not made its decision after the first $O(\log \mathsf{deg}(v) + \log 1/\eps)$ rounds is at most $\eps$. Furthermore, this holds even if the bits of randomness outside the $2$-hops neighborhood of $v$ are determined adversarially.
\end{theorem}
The perhaps surprising fact that the bound only depends on the degree of node $v$, even allowing its neighbors to have infinite degree, demonstrates the \emph{truly local} nature of this algorithm. 
The logarithmic degree-dependency in the bound is optimal, following a lower bound of Kuhn, Moscibroda and Wattenhofer~\cite{kuhn2004localLB}: As indicated by \cite{Fabian}, with minor changes in the arguments of~\cite{kuhn2004localLB}, one can prove that there are graphs in which, the time till each node $v$ can know if it is in MIS or not with constant probability is at least $\Omega(\log \Delta)$ rounds. 

Finally, we note that the fact that \emph{the proof has a locality of $2$-hops}---meaning that the analysis only looks at the $2$-hops neighborhood and particularly, that the guarantee relies only on the coin tosses within the $2$-hops neighborhood of node $v$---will prove vital as we move to global complexity. This might be interesting for practical purposes as well. 

\subsection{Global Complexity}
\label{subsec:global}
Notice that \Cref{thm:local} easily recovers the standard result that after $O(\log n)$ rounds, w.h.p., all nodes have terminated, but now with a local analysis. In light of the $\Omega(\min\{\log \Delta, \sqrt{\log n}\})$ lower bound of Kuhn et al.~\cite{kuhn2004localLB}, it is interesting to find the best possible upper bound, specially when $\log \Delta = o(\log n)$. The best known bound prior to this work was $O(\log^2 \Delta) + 2^{O(\sqrt{\log \log n})}$ rounds, due to Barenboim et al.\cite{barenboim2012locality}.

The overall plan is based on the following nice and natural intuition, which was used in the MIS results of Alon et al.\cite{alon2012LCA} and Barenboim et al.\cite{barenboim2012locality}. We note that this general strategy is often attributed to Beck, as he used it first in his breakthrough algorithmic version of the Lov{\'a}sz Local Lemma\cite{beck1991LLL}. Applied to MIS, the intuition is that, when we run any of the usual randomized MIS algorithms, nodes get removed probabilistically more and more over time. If we run this \emph{base algorithm} for a certain number of rounds, a \emph{graph shattering} type of phenomena occurs. That is, after a certain time, what remains of the graph is a number of ``small'' components, where small might be in regard to size, (weak) diameter, the maximum size of some specially defined independent sets, or some other measure. Once the graph is shattered, one switches to a deterministic algorithm to \emph{finish off} the problem in these remaining small components. 

Since we are considering graphs with max degree $\Delta$, even ignoring the troubling probabilistic dependencies (which are actually rather important), a simplistic intuition based on \emph{Galton-Watson branching processes} tells us that the graph shattering phenomena starts to show up around the time that the probability $\eps$ of each node being left falls below $1/\Delta$\footnote{In truth, the probability threshold is $1/\poly(\Delta)$, because of some unavoidable dependencies. But due to the exponential concentration, the time to reach the $1/\poly(\Delta)$ threshold is within a constant factor of that of the $1/\Delta$ threshold. We will also need to establish some independence, which is not discussed here. See \Cref{sec:global}.}. Alon et al.\cite{alon2012LCA} used an argument of Parnas and Ron~\cite{parnas2007approximating}, showing that Luby's algorithm reaches this threshold after $O(\Delta \log \Delta)$ rounds. Barenboim et al.\cite{barenboim2012locality} used a variant of Luby's, with a small but clever modification, and showed that it reaches the threshold after $O(\log^2 \Delta)$ rounds. As Barenboim et al.\cite{barenboim2012locality} show, after the shattering, the remaining pieces can be solved deterministically, via the help of known deterministic MIS algorithms (and some other ideas), in $\log \Delta \cdot 2^{O(\sqrt{\log \log n})}$ rounds. Thus, the overall complexity of \cite{barenboim2012locality} is $O(\log^2 \Delta) + \log \Delta \cdot 2^{O(\sqrt{\log \log n})} = O(\log^2 \Delta) + 2^{O(\sqrt{\log \log n})}$. 

To improve this, instead of Luby's, we use our new MIS algorithm as the base, which as \Cref{thm:local} suggests, reaches the shattering threshold after only $O(\log \Delta)$ rounds. This will be formalized in \Cref{sec:global}. We will also use some minor modifications for the \emph{post-shattering} phase to reduce it's complexity from $\log \Delta \cdot 2^{O(\sqrt{\log \log n})}$ to $2^{O(\sqrt{\log \log n})}$. The overall result thus becomes:
\begin{theorem}
\label{thm:global} There is a randomized distributed MIS algorithm that terminates after $O(\log\Delta) + 2^{O(\sqrt{\log \log n})}$ rounds, with probability at least $1-1/n$.
\end{theorem}
This improves the best-known bound for MIS and gets close to the $\Omega(\min\{\log \Delta, \sqrt{\log n}\})$ lower bound of Kuhn et al.\cite{kuhn2004localLB}, which at the very least, shows that the upper bound is provably optimal when $\log \Delta \in [2^{\sqrt{\log\log n}}, \sqrt{\log n}]$. Besides that, the new result matches the lower bound in a stronger and much more instructive sense: as we will discuss in point (C2) below, it perfectly pinpoints why the current lower bound techniques cannot prove a lower bound better than $\Omega(\min\{\log \Delta, \sqrt{\log n}\})$.
\subsection{Other Implications}
\label{subsec:implications}

Despite its extreme simplicity, the new algorithm turns out to lead to several implications, when combined with some known results and/or techniques:
\begin{enumerate}
\item[(C1)] Combined with the finish-off phase results of Barenboim et al.\cite{barenboim2012locality}, we get MIS algorithms with complexity $O(\log \Delta) + O(\min\{\lambda^{1+\eps}+\log \lambda \log\log n, \lambda+ \lambda^{\eps} \log\log n, \lambda+ (\log \log n)^{1+\eps}  \})$ for graphs with arboricity $\lambda$. Moreover, combined with the low-arboricity to low-degree reduction of Barenboim et al.\cite{barenboim2012locality}, we get an MIS algorithm with complexity $O(\log \lambda + \sqrt{\log n})$. These bounds improve over some results of \cite{barenboim2012locality}, Barenboim and Elkin\cite{barenboim2010sublogarithmic}, and Lenzen and Wattenhofer\cite{lenzen2011mis}. 

\item[(C2)] \label{tightness} The new results highlight the barrier of the current lower bound techniques. In the known locality-based lower bound arguments, including that of~\cite{kuhn2004localLB}, to establish a $T$-round lower bound, it is necessary that within $T$ rounds, each node sees only a tree. That is, each $T$-hops neighborhood must induce a tree, which implies that the girth must be at least $2T+1$. Since any $g$-girth graph has arboricity $\lambda \leq O(n^{\frac{2}{g-2}})$, from (C1), we get an $O(\sqrt{\log n})$-round MIS algorithm when $g=\Omega(\sqrt{\log n})$. More precisely, for any graph with girth $g=\Omega(\min\{\log \Delta, \sqrt{\log n}\})$, we get an $O(\min\{\log \Delta + 2^{O(\sqrt{\log \log n})}, \sqrt{\log n}\})$-round algorithm. Hence, the $\Omega(\min\{\log \Delta, \sqrt{\log n}\})$ lower bound of~\cite{kuhn2004localLB} is essentially the best-possible when the the topology seen by each node within the allowed time must be a tree. This means, to prove a better lower bound, one has to part with these \emph{``tree local-views"} topologies. However, that gives rise to intricate challenges and actually, to the best of our knowledge, there is no distributed locality-based lower bound, in fact for any (local) problem, that does not rely on \emph{tree local-views}.


\item[(C3)] We get an $O(\sqrt{\log n})$-round MIS algorithm for Erd\"{o}s-R\'{e}nyi random graphs $G(n, p)$. This is because, if $p =\Omega(\frac{2^{\sqrt{\log n}}}{n})$, then with high probability the graph has diameter $O(\sqrt{\log n})$ hops (see e.g. \cite{chung2001diameter}) and when $p =O(\frac{2^{\sqrt{\log n}}}{n})$, with high probability, $\Delta = O(2^{\sqrt{\log n}})$ and thus, the algorithm of \Cref{thm:global} runs in at most $O(\sqrt{\log n})$ rounds.

\item[(C4)] Combined with a recursive sparsification method of Bisht et al.\cite{bisht2014brief}, we get a $(2, \beta)$-ruling-set algorithm with complexity $O(\beta \log^{1/\beta} \Delta) + 2^{O(\sqrt{\log\log n})}$, improving on the complexities of \cite{barenboim2012locality} and \cite{bisht2014brief}. An $(\alpha, \beta)$-ruling set $S$ is a set where each two nodes in $S$ are at distance at least $\alpha$, and each node $v\in V\setminus S$ has a node in $S$ within its $\beta$-hops. So, a $(2, 1)$-ruling-set is simply an MIS. The term $O(\beta \log^{1/\beta} \Delta)$ is arguably (and even \emph{provably}, as \cite{Fabian} indicated,) best-possible for the current method, which roughly speaking works by computing the ruling set iteratively using $\beta$ successive reductions of the degree. 

\item[(C5)] In the Local Computation Algorithms (LCA) model of Rubinfeld et al.\cite{rubinfeld2011fast} and Alon et al.\cite{alon2012LCA}, we get improved bounds for computing MIS. Namely, the best-known time and space complexity improve from, respectively, $2^{O(\log^3 \Delta)} \log^3 n$ and $2^{O(\log^3 \Delta)} \log^2 n$ bounds of Levi, Rubinfeld and Yodpinyanee~\cite{levi2015local} to $2^{O(\log^2 \Delta)} \log^3 n$ and $2^{O(\log^2 \Delta)} \log^2 n$.

\item[(C6)] We get a Weak-MIS algorithm with complexity $O(\log \Delta)$, which thus improves the round complexity of the distributed algorithmic version of the Lov\'{a}sz Local Lemma presented by Chung, Pettie, and Su~\cite{chung2014LLL} from $O(\log_{\frac{1}{ep(\Delta+1)}} n \cdot \log^2 \Delta)$ to  $O(\log_{\frac{1}{ep(\Delta+1)}} n \cdot \log \Delta)$. Roughly speaking, a Weak-MIS computation should produce an independent set $S$ such that for each node $v$, with probability at least $1-1/\poly(\Delta)$, $v$ is either in $S$ or has a neighbor in $S$. 

\item[(C7)] We get an $O(\log \Delta + \log \log \log n)$-round MIS algorithm for the $\mathsf{CONGESTED}$-$\mathsf{CLIQUE}$ model where per round, each node can send $O(\log n)$-bits to each of the other nodes (even those non-adjacent to it): After running the MIS algorithm of \Cref{thm:local} for $O(\log \Delta)$ rounds, w.h.p., if $\Delta \geq n^{0.1}$, we are already done, and otherwise, as \Cref{lem:shattering} shows, all leftover components have size $o(n^{0.5})$. In the latter case, using the algorithm of \cite{fastMST-congestclique}, we can make all nodes know the \emph{leader} of their component in $O(\log \log \log n)$ rounds, and using Lenzen's routing\cite{lenzen2013route}, we can make each leader learn the topology of its whole component, solve the related MIS problem locally, and send back the answers, all in $O(1)$ rounds.  

\end{enumerate}

\section{Warm Up: Local Analysis of Luby's Algorithm}
\label{warmup}
As a warm up for the MIS algorithm of the next section, here, we briefly review Luby's algorithm and present some local analysis for it. The main purpose is to point out the challenge in (tightly) analyzing the local complexity of Luby's, which the algorithm of the next section tries to bypass.

\paragraph{Luby's Algorithm} The algorithm of \cite{luby1985simple, alon1986fast} is as simple and clean as this: \bigskip

\vspace{-10pt}
{\centering
\begin{enumerate} 
\item[] ``\textit{In each \texttt{round}, each node picks a random number\footnote{One can easily see that a precision of $O(\log \Delta)$ bits suffices.} uniformly from $[0, 1]$; strict local \\minimas join the MIS, and get removed from the graph along with their neighbors.}"
\end{enumerate}
}

\noindent Note that each round of the algorithm can be easily implemented in $2$ communication rounds on $G$, one for exchanging the random numbers and the other for informing neighbors of newly joined MIS nodes. Ignoring this 2 factor, in the sequel, by \emph{round} we mean one round of the algorithm.

\paragraph{Global Analysis} The standard method for analyzing Luby's algorithm goes via looking at the whole graph, i.e., using a \emph{global view}. See \cite[Section 12.3]{motwani-raghavan}, \cite[Section 8.4]{peleg:2000}, \cite[Section 4.5]{lynch1996distributed} for textbook treatments. We note that this is the only known way for proving that this algorithm terminates everywhere in $O(\log n)$ rounds with high probability. The base of the analysis is to show that per iteration, in expectation, at least half of the edges (of the whole remaining graph) get removed. Although the initial arguments in \cite{luby1985simple, alon1986fast} were more lengthy, Yves et al.\cite{YvesMIS} pointed out a much simpler argument for this. See \Cref{Yves}, which describes (a paraphrased version of) their argument. By Markov's inequality, this per-round halving implies that after $O(\log n)$ rounds, the algorithm terminates everywhere, with high probability.

\subsection{Local Analysis: Take 1}
To analyze the algorithm in a local way, and to bound its local complexity, the natural idea is to say that over time, each local neighborhood gets ``simplified". Particularly, the first-order realization of this intuition would be to look at the degrees and argue that they shrink with time. The following standard observation is the base tool in this argument:  

\begin{claim}\label{clm:degree-drop} Consider a node $u$ at a particular round, let $d(u)$ be its degree and $d_{max}$ be the maximum degree among the nodes in the inclusive neighborhood $N^{+}(u)$ of $u$. The probability that $u$ is removed in this round is at least $\frac{d(u)+1}{d(u)+d_{max}}$.  
\end{claim}
\begin{proof} Let $u^*$ be the node in $N^{+}(u)$ that draws the smallest random number. If $u^*$ actually has the smallest in its own neighborhood, then it will join MIS which means $u$ gets removed. Since all numbers are iid random variables, and as $u^*$ is the smallest number of $d(u)+1$ of them, the probability that it is the smallest both in its own neighborhood and the neighborhood of $u$ is at least $\frac{d(u)+1}{d(u)+d_{max}}$. This is because, the latter is a set of size at most $d(u)+d_{max}$.
\end{proof}

From the claim, we get that if the degree of a node $u$ is at least half of that of the max of its neighbors, then in one round, with probability at least $1/3$, $u$ gets removed. Thus, in $\alpha=O(1)$ rounds from the start, either $u$ is removed or its degree falls below $\Delta/2$, with probability at least $1/2$. We would like to continue this argument and say that every $O(1)$ rounds, $u$'s degree shrinks by another 2 factor, thus getting a bound of $O(\log \Delta)$. However, this is not straightforward as $u$'s degree drops might get delayed because of delays in the degree drops of $u$'s neighbors. The issue seems rather severe as the degree drops of different nodes can be positively correlated.

Next, we explain a simple argument giving a weak but still local complexity of $O(\log^{2.5} \Delta + \log \Delta \log 1/\eps)$ rounds: For the purpose of this paragraph, let us say a removed node has degree $0$. From above, we get that after $10\alpha\log^{1.5} \Delta$ rounds, the probability that $u$ still has degree at least $\Delta/2$ is at most $2^{-10\log^{1.5} \Delta}$. Thus, using a union bound, we can say that with probability at least $1-(\Delta+1) 2^{-10\log^{1.5} \Delta}$, after $10\alpha\log^{1.5} \Delta$ rounds, $u$ and all its neighbors have degree at most $\Delta/2$. Hence, with probability at least $1-(\Delta+2) 2^{-10\log^{1.5} \Delta}$, after $20\alpha\log^{1.5} \Delta$ rounds, node $u$ has another drop and its degree is at most $\Delta/4$. Continuing this argument pattern recursively for $\log^{0.5} \Delta$ iterations, we get that with probability at least $1-(\Delta+2)^{\log^{0.5} \Delta} \;\cdot\; 2^{-10\log^{1.5} \Delta} \geq 1-2^{-5\log^{1.5} \Delta}$, after $10\alpha\log^{2} \Delta$ rounds, node $u$'s degree has dropped to $\Delta/2^{\log^{0.5} \Delta}$. Now, we can repeat a similar argument, but in blocks of $10\alpha\log^{2} \Delta$ rounds, and each time expecting a degree drop of $2^{\log^{0.5} \Delta}$ factor. We will be able to afford to continue this for $\log^{0.5} \Delta$ iterations and say that, after $10\alpha\log^{2.5} \Delta$ rounds, with probability at least $1-(\Delta+2)^{\log^{0.5} \Delta} \;\cdot\; 2^{-5\log^{1.5} \Delta} \geq 1-2^{-\log^{1.5} \Delta}$, the degree of $u$ has dropped to $1/2$. Since a degree less than $1/2$ means degree $0$, which in turn implies that $v$ is removed, we get that $v$ is removed after at most $O(\log^{2.5} \Delta)$ rounds with probability at least $1-2^{-\Omega(\log^{1.5} \Delta)}$. A simple repetition argument proves that this generalizes to show that after $O(\log^{2.5} \Delta + \log \Delta \log 1/\eps)$ rounds, node $u$ is removed with probability at least $1-\eps$.

In the full version of this paper, we will present a stronger (but also much more complex) argument which proves a local complexity of $O(\log^2 \Delta + \log 1/\eps)$ for the same algorithm. This bound has the desirable additive $\log 1/\eps$ dependency on $\eps$ but it is still far from the best possible bound, due to the first term.

\subsection{Local Analysis: Take 2}
Here, we briefly explain the modification of Luby's algorithm that Barenboim et al.\cite{barenboim2012locality} use. The key is the following clever idea: they \emph{manually} circumvent the problem of nodes having a lag in their degree drops, that is, they \emph{kick out} nodes that their degree drops is lagging significantly out of the algorithm, as these nodes can create trouble for other nodes in their vicinity. 

Formally, they divide time into phases of $\Theta(\log \log \Delta +\log 1/\eps)$ rounds and require that by the end of phase $k$, each node has degree at most $\Delta/2^{k}$. At the end of each phase, each node that has a degree higher than the allowed threshold is \emph{kicked out}. The algorithm is run for $\log \Delta$ phases. From \Cref{clm:degree-drop}, we can see that the probability that a node that has survived up to phase $i-1$ gets \emph{kicked out} in phase $i$ is at most $2^{- \Theta(\log \log \Delta +\log 1/\eps)} = \frac{\eps}{\log \Delta}$. Hence, the probability that a given node $v$ gets kicked out in one of the $\log \Delta$ phases is at most $\eps$. This means, by the end of $\Theta(\log \Delta \log \log \Delta + \log \Delta \log 1/\eps)$ rounds, with probability $1-\eps$, node $v$ is not kicked out and is thus removed because of having degree $0$. That is, it joined or has a neighbor in the MIS. 

This $\Theta(\log \Delta \log \log \Delta + \log \Delta \log 1/\eps)$ local complexity has an improved $\Delta$-dependency (and the guarantee has some nice independence type of properties). However, as mentioned in \Cref{subsec:localIntro}, its $\eps$-dependency is not desirable, due to the $\log \Delta$ factor. Note that this is exactly the reason that the \emph{shattering threshold} in the result of Barenboim et al.\cite{barenboim2012locality} is $O(\log^2 \Delta)$ rounds.  

\section{The New Algorithm and Its Local Complexity}
Here we present a very simple and clean algorithm that guarantees for each node $v$ that after $O(\log \Delta + \log 1/\eps)$  rounds, with probability at least $1-\eps$, node $v$ has terminated and it knows whether it is in the (computed) MIS or it has a neighbor in the (computed) MIS.

\paragraph{The Intuition} Recall that the difficulty in locally analyzing Luby's algorithm was the fact that the degree-dropping progresses of a node can be delayed by those of its neighbors, which in turn can be delayed by their own neighbors, and so on (up to $\log \Delta$ hops). To bypass this issue, the algorithm presented here tries to completely disentangling the ``progress" of node $v$ from that of nodes that are far away, say those at distance above $3$.

The intuitive base of the algorithm is as follows. There are two scenarios in which a node $v$ has a good chance of being removed: either (1) $v$ is trying to join MIS and it does not have too many competing neighbors, in which case $v$ has a shot at joining MIS, or (2) a large enough number of neighbors of $v$ are trying to join MIS and each of them does not have too much competition, in which case it is likely that one of these neighbors of $v$ joins the MIS and thus $v$ gets removed. These two cases also depend only on $v$'s 2-neighborhood. Our key idea is to create an essentially deterministic \emph{dynamic} which has these two scenarios as its (more) stable points and makes each node $v$ spend a significant amount of time in these two scenarios, unless it has been removed already. 

\begin{mdframed}[hidealllines=false,backgroundcolor=gray!30]
\vspace{-5pt}
\paragraph{The Algorithm}
In each round $t$, each node $v$ has a \emph{desire-level} $p_t(v)$ for joining MIS, which initially is set to $p_0(v)=1/2$. We call the total sum of the desire-levels of neighbors of $v$ it's \emph{effective-degree} $d_{t}(v)$, i.e., $d_t(v)=\sum_{u \in N(v)} p_{t}(u)$. The desire-levels change over time as follows: $$p_{t+1}(v)= 
\begin{cases}
    p_{t}(v)/2, & \text{if } d_{t}(v)\geq 2\\
    \min\{2p_{t}(v), 1/2\},  &\text{if } d_{t}(v)< 2.
\end{cases}
$$
The desire-levels are used as follows: In each round, node $v$ gets \emph{marked} with probability $p_{t}(v)$ and if no neighbor of $v$ is marked, $v$ joins the MIS and gets removed along with its neighbors\footnotemark. 
\end{mdframed}

\addtocounter{footnote}{0}\footnotetext{There is a version of Luby's algorithm which also uses a similar marking process. However, at each round, letting $deg(v)$ denote the number of the neighbors of $v$ remaining at that time, Luby's sets the marking probability of each node $v$ to be $\frac{1}{deg(v)+1}$, which by the way is the same as the probability of $v$ being a local minima in the variant described in \Cref{warmup}. Notice that this is a very strict fixing of the marking probability, whereas in our algorithm, we change the probability dynamically/flexibly over time, trying to push towards the two desirable scenarios mentioned in the intuition, and in fact, this simple dynamic is the key ingredient of the new algorithm.}

Again, each round of the algorithm can be implemented in $2$ communication rounds on $G$, one for exchanging the desire-levels and the marks, and the other for informing neighbors of newly joined MIS nodes. Ignoring this 2 factor, in the sequel, each round means a round of the algorithm. 

\smallskip
\paragraph{The Analysis} The algorithm is clearly correct meaning that the set of nodes that join the MIS is indeed an independent set and the algorithm terminates at a node only if the node is either in MIS or adjacent to a node in MIS. We next argue that each node $v$ is likely to terminate quickly.

\begin{theorem} \label{thm:local-restate} For each node $v$, the probability that $v$ has not made its decision within the first $\beta(\log \mathsf{deg} + \log 1/\eps)$ rounds, for a large enough constant $\beta$ and where $\mathsf{deg}$ denotes $v$'s degree at the start of the algorithm, is at most $\eps$. Furthermore, this holds even if the outcome of the coin tosses outside $N^{+}_{2}(v)$ are determined adversarially.
\end{theorem}
Let us say that a node $u$ is \emph{low-degree} if $d_t(u)<2$, and \emph{high-degree} otherwise. Considering the intuition discussed above, we define two types of \emph{golden rounds} for a node $v$: (1) rounds in which $d_t(v)<2$ and $p_{t}(v)= 1/2$, (2) rounds in which $d_{v}(t)\geq 1$ and at least $d_{t}(v)/10$ of it is contributed by low-degree neighbors. These are called golden rounds because, as we will see, in the first type, $v$ has a constant chance of joining MIS and in the second type there is a constant chance that one of those low-degree neighbors of $v$ joins the MIS and thus $v$ gets removed. For the sake of analysis, let us imagine that node $v$ keeps track of the number of golden rounds of each type it has been in.

\begin{lemma}\label{lem:goldCount} By the end of round $\beta(\log \mathsf{deg} + \log 1/\eps)$, either $v$ has joined, or has a neighbor in, the (computed) MIS, or at least one of its golden round counts reached $100 (\log \mathsf{deg} + \log 1/\eps)$.
\end{lemma}
\begin{proof} We focus only on the first $\beta(\log \mathsf{deg} + \log 1/\eps)$ rounds. Let $g_1$ and $g_2$ respectively be the number of golden rounds of types 1 and 2 for $v$, during this period. We assume that by the end of round $\beta(\log \mathsf{deg}+ \log 1/\eps)$, node $v$ is not removed and $g_1 \leq 100 (\log \mathsf{deg} + \log 1/\eps)$, and we conclude that, then it must have been the case that $g_2 > 100 (\log \mathsf{deg} + \log 1/\eps)$.

Let $h$ be the number of rounds during which $d_{t}(v)\geq 2$. Notice that the changes in $p_{t}(v)$ are governed by the condition  $d_{t}(v)\geq 2$ and the rounds with $d_{t}(v)\geq 2$ are exactly the ones in which $p_{t}(v)$ decreases by a $2$ factor. Since the number of $2$ factor increases in $p_{t}(v)$ can be at most equal to the number of $2$ factor decreases in it, we get that there are at least $\beta(\log \mathsf{deg} + \log 1/\eps) -2h$ rounds in which $p_{t}(v)=1/2$. Now out of these rounds, at most $h$ of them can be when $d_{t}(v)\geq 2$. Hence, $g_1 \geq \beta(\log \mathsf{deg} + \log 1/\eps) -3h$. As we have assumed $g_1 \leq 100 (\log \mathsf{deg} + \log 1/\eps)$, we get that $\beta(\log \mathsf{deg} + \log 1/\eps) -3h \leq 100 (\log \mathsf{deg} + \log 1/\eps)$. Since $\beta\geq 1300$, we get $h\geq 400 (\log \mathsf{deg} + \log 1/\eps)$.

Let us consider the changes in the effective-degree $d_{t}(v)$ of $v$ over time. If $d_{t}(v) \geq 1$ and this is not a golden round of type-2, then we have $$d_{t+1}(v) \leq 2 \frac{1}{10} d_v(t)+ \frac{1}{2} \frac{9}{10} d_{v}(t) < \frac{2}{3} d_{t}(v).$$ There are $g_2$ golden rounds of type-2. Except for these, whenever $d_{t}(v)\geq 1$, the effective-degree $d_{t}(v)$ shrinks by at least a $2/3$ factor. In those exceptions, it increases by at most a $2$ factor. Each of these exception rounds cancels the effect of at most $2$ shrinkage rounds, as $(2/3)^2 \times 2 <1$. Thus, ignoring the total of at most $3g_2$ rounds lost due to type-2 golden rounds and their cancellation effects, every other round with $d_{t}(v)\geq 2$ pushes the effective-degree down by a $2/3$ factor\footnote{Notice the switch to $d_{t}(v)\geq 2$, instead of $d_{t}(v)> 1$. We need to allow a small slack here, as done by switching to threshold $d_{t}(v)\geq 2$, in order to avoid the possible zigzag behaviors on the boundary. This is because, the above argument does not bound the number of $2$-factor increases in $d_{t}(v)$ that start when $d_{t}(v)\in (1/2, 1)$ but these would lead $d_{t}(v)$ to go above $1$. This can continue to happen even for an unlimited time  if $d_{t}(v)$ keeps zigzagging around $1$ (unless we give further arguments of the same flavor showing that this is not possible). However, for $d_{t}(v)$ to go/stay above $2$, it takes increases that start when $d_{t}(v)>1$, and the number of these is upper bounded to $g_2$.}. This cannot (continue to) happen more than $\log_{3/2} \mathsf{deg}$ often as that would lead the effective degree to exit the $d_{t}(v)\geq 2$ region. Hence, the number of rounds in which $d_{t}(v)\geq 2$ is at most $\log_{3/2} \mathsf{deg} + 3g_2$. That is, $h \leq \log_{3/2} \mathsf{deg} + 3g_2$. Since $h\geq 400 (\log \mathsf{deg} + \log 1/\eps)$, we get $g_2 > 100 (\log \mathsf{deg} + \log 1/\eps)$.
\end{proof}

\begin{lemma} In each type-1 golden round, with probability at least $1/100$, $v$ joins the MIS. Moreover, in each type-2 golden round, with probability at least $1/100$, a neighbor of $v$ joins the MIS.  Hence, the probability that $v$ has not been removed (due to joining or having a neighbor in MIS) during the first $\beta(\log \mathsf{deg} + \log 1/\eps)$ rounds is at most $\eps$. These statements hold even if the coin tosses outside $N^+_{2}(v)$ are determined adversarially.
\end{lemma}
\begin{proof}
In each type-1 golden round, node $v$ gets marked with probability $1/2$. The probability that no neighbor of $v$ is marked is $\prod_{u \in N(v)} (1-p_t(u)) \geq 4^{-\sum_{u \in N(v)} p_t(v)} = 4^{-d_{t}(v)} > 4^{-2}=1/16$. Hence, $v$ joins the MIS with probability at least $1/32>1/100$.

Now consider a type-2 golden round. Suppose we walk over the set $L$ of low-degree neighbors of $v$ one by one and expose their randomness until we reach a node that is marked. We will find a marked node with probability at least $$1-\prod_{u \in \textit{L}} (1-p_{u}(t)) \geq 1- e^{-\sum_{u \in \textit{L}} p_{u}(t)} \geq 1-e^{-d_{t}(v)/10} \geq 1-e^{-1/10} >0.08.$$ When we reach the first low-degree  neighbor $u$ that is marked, the probability that no neighbor of $u$ gets marked is at least $\prod_{w\in N(u)} (1-p_{t}(w)) \geq 4^{-\sum_{w \in N(u)} p_t(w)} \geq 4^{-d_{t}(u)} > 1/16$. Hence, with probability at least $0.08/16=1/100$, one of the neighbors of $v$ joins the MIS. 

We now know that in each golden round, $v$ gets removed with probability at least $1/100$, due to joining MIS or having a neighbor join the MIS. Thus, using \Cref{lem:goldCount}, we get that the probability that $v$ does not get removed is at most $(1-1/100)^{100(\log \mathsf{deg} + \log 1/\eps)} \leq \eps/\mathsf{deg} \leq \eps$.
\end{proof}

\section{Improved Global Complexity}
\label{sec:global}
In this section, we explain how combining the algorithm of the previous section with some known techniques leads to a randomized MIS algorithm with a high probability global complexity of $O(\log \Delta) + 2^{O(\sqrt{\log \log n})}$ rounds. 

As explained in \Cref{subsec:global}, the starting point is to run the algorithm of the previous section for $\Theta(\log \Delta)$ rounds. Thanks to the local complexity of this base algorithm, as we will show, we reach the \emph{shattering threshold} after $O(\log \Delta)$ rounds. The $2$-hops \emph{randomness locality} of \Cref{thm:local-restate}, the fact that it only relies on the randomness bits within $2$-hops neighborhood, plays a vital role in establishing this shattering phenomena. The precise statement of the shattering property achieved is given in \Cref{lem:shattering}, but we first need to establish a helping lemma:

\begin{lemma}\label{lem:smallComp} Let $c>0$ be an arbitrary constant. For any $5$-independent set of nodes $S$---that is, a set in which the pairwise distances are at least $5$---the probability that all nodes of $S$ remain undecided after $\Theta(c\log \Delta)$ rounds of the MIS algorithm of the previous section is at most $\Delta^{-c|S|}$.
\end{lemma}
\begin{proof} We walk over the nodes of $S$ one by one: when considering node $v\in S$, we know from that \Cref{thm:local-restate} that the probability that $v$ stays undecided after $\Theta(c\log \Delta)$ rounds is at most $\Delta^{-c}$, and more importantly, this only relies on the coin tosses within distance $2$ of $v$. Because of the $5$-independence of set $S$, the coin tosses we rely on for different nodes of $S$ are non-overlapping and hence, the probability that the whole set $S$ stays undecided is at most $\Delta^{-c|S|}$.
\end{proof}

From this lemma, we can get the following \emph{shattering} guarantee. Since the proof is similar to that of \cite[Lemma 3.3]{barenboim2012locality}, or those of \cite[Main Lemma]{beck1991LLL}, \cite[Lemma 4.6]{alon2012LCA}, and \cite[Theorem 3]{levi2015local}, we only provide a brief sketch:
\begin{lemma}
\label{lem:shattering} Let $c$ be a large enough constant and $B$ be the set of nodes remaining undecided after $\Theta(c\log \Delta)$ rounds of the MIS algorithm of the previous section on a graph $G$. Then, with probability at least $1-1/n^{c}$, we have the following two properties:
\begin{enumerate}
\item[(P1)] There is no $(G^{4^-})$-independent $(G^{9^-})$-connected subset $S \subseteq B$ s.t. $|S|\geq \log_{\Delta} n$. Here $G^{x^-}$ denotes the graph where we put edges between each two nodes with $G$-distance at most $x$.

\item[(P2)] All connected components of $G[B]$, that is the subgraph of $G$ induced by nodes in $B$, have each at most $O(\log_{\Delta} n \cdot \Delta^4)$ nodes.
\end{enumerate}
\end{lemma}
\begin{proof}[Proof Sketch] Let $H=G^{9^-} \setminus G^{4^-}$, i.e., the result of removing $G^{4-}$ edges from $G^{9-}$. For (P1), note that the existence of any such set $S$ would mean $H[B]$ contains a $(\log_{\Delta} n)$-node tree subgraph. There are at most $4^{\log_{\Delta} n}$ different $(\log_{\Delta} n)$-node tree topologies and for each of them, less than $n \Delta^{\log_{\Delta} n}$ ways to embed it in $H$. For each of these trees, by \Cref{lem:smallComp}, the probability that all of its nodes stay is at most $\Delta^{-c(\log_{\Delta} n)}$. By a union bound over all trees, we conclude that with probability $1-n (4\Delta)^{ \log_{\Delta} n} \Delta^{-c(\log_{\Delta} n)} \geq 1-1/n^{c}$, no such such set $S$ exists. For (P2), note that if $G[B]$ has a component with more than $\Theta(\log_{\Delta} n \cdot \Delta^4)$ nodes, then we can find a set $S$ violating (P1): greedily add nodes to the candidate $S$ one-by-one, and each time discard all nodes within $4$-hops of the newly added node, which are at most $O(\Delta^4)$ many.
\end{proof}

From property (P2) of \Cref{lem:shattering}, it follows that running the deterministic MIS algorithm of Panconesi and Srinivasan\cite{panconesi1992improved}, which works in $2^{O(\log n')}$ rounds in graphs of size $n'$, in each of the remaining components finishes our MIS problem in $2^{O(\sqrt{\log \Delta + \log \log n}})$ rounds. However, the appearance of the $\log \Delta$ in the exponent is undesirable, as we seek a complexity of $O(\log \Delta) + 2^{O(\sqrt{\log \log n}})$. To remedy this problem, we use an idea similar to \cite[Section 3.2]{barenboim2012locality}, which tries to leverage the (P1) property. 

In a very rough sense, the (P1) property of \Cref{lem:shattering} tells us that if we ``contract nodes that are closer than 5-hops" (this is to be made precise), the left over components would have size at most $\log_{\Delta} n$, which would thus avoid the undesirable $\log \Delta$ term in the exponent. We will see that, while running the deterministic MIS algorithm, will be able to expand back these contractions and solve their local problems. We next formalize this intuition. 

The \emph{finish-off} algorithm is as follows: We consider each connected component $C$ of the remaining nodes separately; the algorithm runs in parallel for all the components. First compute a $(5, h)$-ruling set $R_{C}$ in each connected component $C$ of the set $B$ of the remaining nodes, for an $h=\Theta(\log \log n)$. Recall that a $(5, h)$-ruling set $R_{C}$ means each two nodes of $R_{C}$ have distance at least $5$ while for each node in $C$, there is at least one node in $R_C$ within its $h$-hops. This $(5, h)$-ruling set $R_C$ can be computed in $O(\log \log n)$ rounds using the algorithm\footnote{This is different than what Barenboim et al. did. They could afford to use the more standard ruling set algorithm, particularly computing a $(5, 32\log \Delta+O(1))$-ruling set for their purposes, because the fact that this $32 \log \Delta$ ends up multiplying the complexity of their finish-off phase did not change (the asymptotics of) their overall complexity.} of Schneider, Elkin and Wattenhofer\cite{schneider2013symmetry}. See also\cite[Table 4]{barenboim2012locality}. Form clusters around $R_{C}$-nodes by letting each node $v\in C$ join the cluster of the nearest $R_{C}$-node, breaking ties arbitrarily by IDs. Then, contract each cluster to a new node. Thus, we get a new graph $G'_C$ on these new nodes, where in reality, each of these new nodes has radius $h=O(\log \log n)$ and thus, a communication round on $G'_C$ can be simulated by $O(h)$ communication rounds on $G$. 

From (P1) of \Cref{lem:smallComp}, we can infer that $G'_C$ has at most $\log_\Delta n$ nodes, w.h.p., as follows: even though $R_C$ might be disconnected in $G^{9-}$, by greedily adding more nodes of $C$ to it, one by one, we can make it connected in $G^{9-}$ but still keep it $5$-independent. We note that this is done only for the analysis. See also \cite[Page 19, Steps 3 and 4]{barenboim2012locality} for a more precise description. Since by (P1) of \Cref{lem:smallComp}, the end result should have size at most $\log_\Delta n$, with high probability, we conclude $G'_C$ has at most $\log_\Delta n$ nodes, with high probability.

We can now compute an MIS of $C$, via almost the standard deterministic way of using network decompositions. We run the network decomposition algorithm of Panconesi and Srinivasan\cite{panconesi1992improved} on $G'_C$. This takes $2^{O(\sqrt{\log \log_\Delta n})} $ rounds and gives $G'_C$-clusters of radius at most $2^{O(\sqrt{\log \log_\Delta n})} $, colored with $2^{O(\sqrt{\log \log_\Delta n})}$ colors such that adjacent clusters do not have the same color. We will walk over the colors one by one and compute the MIS of the clusters of that color, given the solutions of the previous colors. Each time, we can (mentally) expand each of these $G'_C$ clusters to all the $C$-nodes of the related cluster, which means these $C$-clusters have radius at most $\log \log n \cdot 2^{O(\sqrt{\log \log_\Delta n})}$. While solving the problem of color-$j$ clusters, we make a node in each of these clusters gather the whole topology of its cluster and also the adjacent MIS nodes of the previous colors. Then, this cluster-center solves the MIS problem locally, and reports it back. Since each cluster has radius $\log \log n \cdot 2^{O(\sqrt{\log \log_\Delta n})}$, this takes $\log \log n \cdot 2^{O(\sqrt{\log \log_\Delta n})}$ rounds per color. Thus, over all the colors, the complexity becomes $2^{O(\sqrt{\log\log_{\Delta} n})} \cdot \log \log n \cdot 2^{O(\sqrt{\log \log_\Delta n})} = 2^{O(\sqrt{\log \log n})}$ rounds. Including the $O(\log \log n)$ ruling-set computation rounds and the $O(\log \Delta)$ pre-shattering rounds, this gives the promised global complexity of $O(\log \Delta) + 2^{O(\sqrt{\log \log n})}$, hence proving \Cref{thm:global}.  


\section{Concluding Remarks}
\vspace{-8pt}
This paper presented an extremely simple randomized distributed MIS algorithm, which exhibits many interesting \emph{local} characteristics, including a \emph{local complexity} guarantee of each node $v$ terminating in $O(\log \mathsf{deg}(v)+\log 1/\eps)$ rounds, with probability at least $1-\eps$. We also showed that combined with known techniques, this leads to an improved high probability global complexity of $O(\log \Delta) + 2^{O(\sqrt{\log \log n})}$ rounds, and several other important implications, as described in \Cref{subsec:implications}.

For open questions, the gap between the upper and lower bounds, which shows up when $\log \Delta = \omega (\sqrt{\log n})$, is perhaps the most interesting. We saw in (C2) of \Cref{subsec:implications} that if the lower-bound is the one that should be improved, we need to go away from ``tree local-views" topologies. Another longstanding open problem is to find a $\poly(\log n)$ deterministic distributed MIS algorithm. Combined with the results of this paper, that can potentially get us to an $O(\log \Delta) + \poly(\log \log n)$ randomized algorithm.  
\medskip

\paragraph{Acknowledgment}
I thank Eli Gafni, Bernhard Haeupler, Stephan Holzer, Fabian Kuhn, Nancy Lynch, and Seth Pettie for valuable discussions. I am also grateful to Fabian Kuhn and Nancy Lynch for carefully reading the paper and  many helpful comments. The point (C2) in \Cref{subsec:implications} was brought to my attention by Fabian Kuhn. The idea of highlighting the \emph{local complexity} is rooted in conversations with Stephan Holzer and Nancy Lynch, and also in Eli Gafni's serious insistence\footnote{This is a paraphrased version of his comment during a lecture on Linial's $\Omega(\log^* n)$ lower bound, in the Fall 2014 Distributed Graph Algorithms (DGA) course at MIT.} that \emph{the (true) complexity of a local problem should not depend on $n$}. 
 
\newpage
\bibliographystyle{alpha}
\bibliography{ref}

\appendix
\section{Simplified Global Analysis of Luby's, due to Yves et al.}
\label{Yves}
We here explain (a slightly paraphrased version of) the clever approach of Yves et al.\cite{YvesMIS} for bounding Luby's global time complexity:

\begin{lemma} \label{lem:Main} Let $G[V_t]$ be the graph induced by nodes that are alive in round $t$, and let $m_t$ denote the number of edges of $G[V_t]$. For each round $t$, we have $\mathbb{E}[m_{t+1}] \leq \frac{m_{t}}{2}$, where the expectation is on the randomness of round $t$. 
\end{lemma}
\begin{proof}
Consider an edge $e=(u, v)$ that is alive at the start of a round $t$, i.e., that is in $G[V_t]$. Note that edge $e$ will not be in $G[V_{t+1}]$, in which case we say $e$ died, if in the random numbers drawn in round $t$, there is a node $w$ that is adjacent to $v$ or $u$ (or both) and $w$ has the strict local minima of its own neighborhood. In this case, we say node $w$ \emph{killed} edge $e$. 

Note that the probability that $w$ kills $e$ is $\frac{1}{d(w)+1}$, where $d(w)$ denotes the degree of $w$. The difficulty comes when we want to compute the probability that \emph{there exists} a $w$ that kills $e$. This is mainly because, the events of different $w$ killing $e$ are not disjoint, and hence we cannot easily sum over them. Fortunately, there is a simple and elegant change in the definition, due to Yves et al.\cite{YvesMIS}, which saves us from tedious calculations:

Without loss of generality, suppose that $w$ is adjacent to $v$. We say $w$ \emph{strongly kills} $e$ from the side of node $v$---and use notation $w\stackrel{v}{\rightarrow}e$ to denote it---if $w$ has the (strictly) minimum random number in $\Gamma_w\cup \Gamma_v$. Note that this is a stronger requirement and thanks to this definition, at most one node $w$ can strongly kill $e$ from the side of $v$. Thus, in a sense, we now have events that are disjoint which means the probability that any of them happens is the summation of the probabilities of each of them happening. The only catch is, we might \emph{double count} an edge dying, because it gets (strongly) killed from both endpoints, but that is easy to handle; we just lose a $2$-factor. In the following, with a slight abuse of notation, by $E$ we mean the alive edges, i.e., those of $G[V_t]$, and by $\Gamma(v)$, we mean the neighbors of $v$ in $G[V_t]$. We have

\begin{align*}
\mathbb{E}[\textit{Number of edges that die}] &\geq& \sum_{e=(v, w)\in E} \Pr[e \textit{ gets strongly killed}]\\ 
&\geq&  \sum_{e=(v, w)\in E} \bigg(\sum_{w\in \Gamma(v)}\Pr[w\stackrel{v}{\rightarrow}e] + \sum_{w'\in \Gamma(u)}\Pr[w'\stackrel{u}{\rightarrow}e]\bigg)/2 \\ 
&\geq& \sum_{e=(v, w)\in E} \bigg(\sum_{w\in \Gamma(v)}\frac{1}{d(w)+d(v)} + \sum_{w'\in \Gamma(u)}\frac{1}{d(w)+d(u)}\bigg)/2 \\
&=& \sum_{v\in V} \sum_{u \in \Gamma(v), \; e=(u, v)} \bigg(\sum_{w\in \Gamma(v)}\frac{1}{d(w)+d(v)} + \sum_{w'\in \Gamma(u)}\frac{1}{d(w)+d(u)}\bigg)/2 \\ 
&=& \bigg(\sum_{v\in V} \sum_{w\in \Gamma(v)} \sum_{u \in \Gamma(v), \; e=(u, v)} \frac{1}{d(w)+d(v)} \\
&+& \sum_{u\in V} \sum_{w'\in \Gamma(u)} \sum_{v \in \Gamma(u), \; e=(u, v)} \frac{1}{d(w')+d(u)}\bigg)/2\\
&=& \bigg(\sum_{v\in V} \sum_{w\in \Gamma(v)} \frac{d(v)}{d(w)+d(v)}+\sum_{u\in V} \sum_{w'\in \Gamma(u)} \frac{d(u)}{d(w')+d(u)} \bigg)/2\\
&=& \bigg(\sum_{v\in V} \sum_{w\in \Gamma(v)} \frac{d(v)}{d(w)+d(v)}+\sum_{v\in V} \sum_{w\in \Gamma(v)} \frac{d(w)}{d(w)+d(v)} \bigg)/2 \\
&=& \bigg(\sum_{v\in V} \sum_{w\in \Gamma(v)} 1\bigg)/2 = m_{t}/2.
\end{align*}
\end{proof}
 
It follows from Lemma \ref{lem:Main} that the expected number of the edges that are alive after $4\log n$ rounds is at most $\frac{n^2 /2}{2^{4\log n}} < \frac{1}{n^2}$. Therefore, using Markov's inequality, we conclude that the probability that there is at least $1$ edge that is left alive is at most $\frac{1}{n^2}$. Hence, with probability at least $1-\frac{1}{n^2}$, all nodes have terminated by the end of round $4\log n$.
\end{document}